\documentclass[11pt]{article}
\pdfoutput=1
\def\withcolors{1}
\def\withnotes{0} 

\usepackage{ccanonne}
\usepackage{palatino}
\usepackage{braket}
\usepackage{booktabs}
\usepackage{colortbl}
\usepackage{comment}

\newcommand{\cyan}[1]{\textcolor{cyan!70!blue!100}{#1}}

\DeclarePairedDelimiter\parens{\lparen}{\rparen}

\DeclareMathOperator*{\avg}{avg}
\newcommand{\E}{\shortexpect}
\newcommand{\wh}[1]{\widehat{#1}}

\newcommand{\bj}{\boldsymbol{j}}

\newcommand{\bJ}{\boldsymbol{J}}

\newcommand{\bX}{\boldsymbol{X}}
\newcommand{\bY}{\boldsymbol{Y}}
\newcommand{\bmu}{\boldsymbol{\mu}}
\newcommand{\bsigma}{\boldsymbol{\sigma}}

\colorlet{tableheadcolor}{gray!25} %
\colorlet{tablerowcolor}{gray!10} %

\title{\vspace{-3ex}\bfseries Uniformity testing when you have the source code\vspace{-0.5ex}}

\author{Cl\'ement L. Canonne\thanks{The University of Sydney.} 
\and Robin Kothari\thanks{Google Quantum AI.} 
\and Ryan O'Donnell\thanks{Carnegie Mellon University. Supported by ARO grant W911NF2110001 and by a gift from Google Quantum~AI.}}
\date{\vspace{-5ex}}

\begin{document}
\maketitle
\begin{abstract}
  We study quantum algorithms for verifying properties of the output probability distribution of a classical or quantum circuit, given access to the source code that generates the distribution. We consider the basic task of uniformity testing, which is to decide if the output distribution is uniform on $[\dims]$ or $\dst$-far from uniform in total variation distance.  More generally, we  consider identity testing, which is the task of deciding if the output distribution equals a known hypothesis distribution, or is $\dst$-far from it. For both problems, the previous best known upper bound was $O(\min\{\dims^{1/3}/\dst^{2},\dims^{1/2}/\dst\})$. 
  Here we improve the upper bound to $O(\min\{\dims^{1/3}/\dst^{4/3}, \dims^{1/2}/\dst\})$, which we conjecture is optimal.
\end{abstract}

\section{Introduction}
For 30 years we have known that quantum computers can solve certain problems significantly faster than any known classical algorithm. 
Traditionally, most of the research in this area has focused on decision problems (like SAT) or function problems (like Factoring), where for each possible input there is a unique ``correct'' output.
However, we have also found that quantum computers can yield speedups for the task of \emph{sampling} from certain probability distributions. Prominent examples include boson sampling~\cite{AA13} and random circuit sampling~\cite{BIS+18}. 
Sampling tasks have seemed more natural for NISQ-era quantum computation, and indeed many of the first candidate experimental demonstrations of quantum advantage have been for sampling problems~\cite{AAB+19}. 

One of the downsides of sampling problems is the challenge of \emph{verifying} the output of an algorithm, whether classical or quantum, that claims to sample from a certain distribution. 
As a simple example, consider a classical or quantum algorithm that implements a supposed hash function with output alphabet $[\dims]\eqdef\{1,\ldots,\dims\}$. The algorithm designer claims that the output distribution of this hash function is uniform on $[\dims]$. %
If $\p$ denotes the actual output distribution of the algorithm, and $\uniform_\dims$ denotes the uniform distribution on $[\dims]$, then we would like to test whether $\p=\uniform_\dims$, and reject the claim if $\p$ is in fact $\dst$-far from $\uniform_\dims$ in total variation distance, meaning $\frac{1}{2} \norm{\p-\uniform_\dims}_1 > \dst$.
(We will also consider other distance measures in this work, since the complexity of the testing task is sensitive to this choice.) 

This verification task is called ``uniformity testing'' (in total variation distance) and its complexity is well studied in the classical literature. If we only have access to samples from $\p$, but are not allowed to inspect the algorithm that produces these samples, it is known that $\Theta(\dims^{1/2}/\dst^2)$ samples are necessary and sufficient to solve this problem; there are various classical algorithms that achieve this bound (starting with that of~\cite{Paninski08}; see, \eg~\cite{CanonneTopicsDT2022} for a detailed survey and discussion), and it is also not possible to do better with a quantum algorithm. But what if---as in the examples above---we \emph{do} have access to the algorithm that produces $\p$? Can we improve on this complexity if we have access to the ``source code'' of the algorithm? 

\paragraph{Having the source code.} To clarify, the ``source code'' for a classical randomized sampling algorithm means a randomized circuit (with no input) whose output is one draw from~$\p$.
More generally, the ``source code'' for a quantum sampling algorithm means a unitary quantum circuit (with all input qubits fixed to~$\ket{0}$) which gives one draw from~$\p$ when some of its output bits are measured in the standard basis and the rest are discarded.\footnote{This is sometimes termed the ``purified quantum query access model'', and is the most natural and general model.  The ``quantum string oracle'', referenced later in \Cref{table:1}, refers to a situation in which one assumes a very specific type of source code for~$\p$ (thus making algorithmic tasks easier). See \Cref{sec:prelims} for details and~\cite{Bel19} for a thorough discussion.}
The simplest way to use the code $C$ for~$\p$ is to run it, obtaining one sample.  
If $C$ has size~$S$, then getting one sample this way has cost~$S$.  
Another way to use the code~$C$ is to deterministically compute all its output probabilities; this gives one perfect information about~$\p$, but has cost bound~$2^S$.  
But quantum computing has suggested a third way to use the code: ``running it in reverse''.  
For example, Grover's original algorithm~\cite{Gro96} can be seen as distinguishing two possibilities for $\p$ on $[2]$, namely $\p_1 = 0$ or $\p_1 = 1/N$, while using only $O(N^{1/2})$ forwards/backwards executions of~$C$.
The total cost here is $O(N^{1/2}) \cdot S$, the same as the cost for $O(N^{1/2})$ samples. 

We suggest that the utility of ``having the source code'' for distribution testing problems remains notably underexplored.
Indeed, there is significant room for improvment in the bounds for even the most canonical of all such problems: uniformity testing.  
Our main theorem is the following:
\begin{theorem}\label{thm:main}
  There is a computationally efficient quantum algorithm for uniformity testing with the following guarantees: given $\dst \geq 1/\sqrt{\dims}$, the algorithm makes $O(\dims^{1/3}/\dst^{4/3})$ uses of  ``the code'' for an unknown distribution $\p$ over $[\dims]$, and distinguishes with probability at least~$.99$ between 
  \begin{equation}
      (1)~\p=\uniform_\dims, \qquad \textrm{and}\qquad (2)~\totalvardist{\p}{\uniform_\dims} > \dst.    
  \end{equation}
\end{theorem}
The main idea behind this theorem is to combine very careful classical probabilistic analysis with a black-box use of Quantum Mean Estimation (QME)~\cite{Gro98,BHT98,Hei02,Mon15,Ham21,KO23}; see \Cref{sec:ideas} for further discussion.
\Cref{table:1} below compares our result to prior work on the problem.
\begin{table}[b]  
  \begin{tabular}{cccc}
  \toprule
  Reference  & Large $\dst$ regime    & Small $\dst$ regime  &  Access model \\
  \midrule
  \cite{Paninski08,AcharyaDK15} \hspace{-1em} & $\Theta(\dims^{1/2}/\dst^2)$ & & Classical, no source code \\
  \cite{BHH11} & $O(\dims^{1/3})$ for $\dst = \Theta(1)^*$  & & Quantum string oracle \\
  \cite{CFMdW10} & $O(\dims^{1/3}/\dst^2)$ & & Quantum string oracle \\
  \cite{GL20} & & \hspace{-2em}$O(\dims^{1/2}/\dst) \cdot \log(\dims/\dst)^3 \log \log(\dims/\dst)$\hspace{-1em} & Source code  \\
  \cite{LuoWL24} & & $O(\dims^{1/2}/\dst)$ & Source code \\
  \bf{}This work  & $O(\dims^{1/3}/\dst^{4/3})$ for $\dst\geq \frac{1}{\sqrt\dims}$ & & Source code \\
  \bottomrule
  \end{tabular}
  {\footnotesize *The work states a bound of $O(\dims^{1/3}/\dst^{4/3})$, but adds that $\dst$ must be constant.}
  \caption{``Sample'' complexity for uniformity testing with respect to total variation distance \label{table:1}}
\end{table}
\Cref{table:1} has two columns because it seems that different algorithms are necessary depending on how $\dims$ and $\dst$ relate.
(Interestingly, this is not the case in the classical no-source-code model.)
Thus combining our new result with that of~\cite{LuoWL24}, the best known upper bound becomes $O(\min\{\dims^{1/3}/\dst^{4/3}, \dims^{1/2}/\dst\})$.
We remark that although~\cite{LuoWL24}'s algorithm/analysis is already simple, we give an alternative simple algorithm and analysis achieving $O(\dims^{1/2}/\dst)$ in \Cref{sec:small}, employing the classical analysis + QME approach used in the proof of our main theorem.

\paragraph{Lower bounds?} As for lower bounds (holding even in the quantum string oracle model): complexity $\Omega(1/\dst)$ is necessary even in the case of constant $\dims = 2$, following from work of \cite{NW99}; and, \cite{CFMdW10} showed a lower bound of $\Omega(\dims^{1/3})$ even in the case of constant~$\dst$, by  reduction from the collision problem~\cite{AS04}.
For reasons discussed in \Cref{sec:ideas}, we make the (somewhat bold) conjecture that our new upper bound is in fact tight for all $\dims$ and~$\dst$:
\begin{conjecture} \label{conj:tv}
    Any algorithm that distinguishes  $\p = \uniform_\dims$ from $\totalvardist{\p}{\uniform_\dims} > \dst$ with success probability at least~$.99$ requires $\Omega(\min\{\dims^{1/3}/\dst^{4/3}, \dims^{1/2}/\dst\})$ uses of the code for~$\p$. (Moreover, we conjecture this lower bound in the stronger quantum string oracle model.)
\end{conjecture}

\paragraph{Identity testing.} 
Several prior works in this area have also studied the following natural generalization of uniformity testing: testing identity of the unknown distribution~$\p$ to a known hypothesis distribution~$\q$.
An example application of this might be when $\q$ is a Porter--Thomas-type distribution arising as the ideal output of a random quantum circuit. 
Luckily, fairly recent work has given a completely generic reduction from \emph{any} fixed identity testing problem to the uniformity testing problem; see~\cite{Goldreich:16}, or~\cite[Section~2.2.3]{CanonneTopicsDT2022}. We can therefore immediately extend our new theorem to the general identity-testing setting:
\begin{corollary}
  There is a computationally efficient quantum algorithm for identity testing to a reference distribution $\q$ over $[\dims]$ with the following guarantees: The algorithm makes $O(\min(\dims^{1/3}/\dst^{4/3}, \dims^{1/2}/\dst))$ uses of  ``the code'' for an unknown distribution $\p$ over $[\dims]$, and distinguishes with probability at least $.99$ between
\begin{equation}
    (1)~\p=\q, \qquad \text{and} \qquad (2)~\totalvardist{\p}{\q} > \dst.    
\end{equation}
\end{corollary}
\noindent(For completeness, we verify in~\cref{app:identity:reduction} that the blackbox reduction does indeed carry through in our setting, preserving access to ``the code''.)

\paragraph{More fine-grained results.}
In proving our main theorem, we will in fact prove a strictly stronger version, one which is more fine-grained in two ways: 

(1)~\emph{Tolerance:} Not only does our test accept with high probability when $\p = \uniform_\dims$, it also accepts with high probability when $\p$ is sufficiently close to $\uniform_\dims$.  

(2)~\emph{Stricter distance measure.} Not only does our test reject with high probability when $\totalvardist{\p}{\uniform_\dims} > \dst$, it also rejects with high probability when $\hellinger{\p}{\uniform_\dims} > \dst$.
(This is stronger, since $\totalvardist{\p}{\q} \leq \hellinger{\p}{\q}$ always.)

To elaborate, recall the below chain of inequalities, which also includes KL- and $\chi^2$-divergence.  
(We review probability distance measures in \Cref{sec:prelims}.)
\begin{equation}  \label{eqn:ineqs}
  \totalvardist{\p}{\q}^2 \leq \hellingersq{\p}{\q} \leq \kldiv{\p}{\q} \leq \chisquare{\p}{\q}.
\end{equation}
The strictly stronger version of \Cref{thm:main} that we prove is:
\begin{theorem}\label{thm:main2}
  There is a computationally efficient quantum algorithm for uniformity testing with the following guarantees: For $1/\dims \leq \thresh \leq 1$, the algorithm makes $O(\dims^{1/3}/\thresh^{2/3})$ uses of  ``the code'' for an unknown distribution $\p$ over $[\dims]$, and distinguishes with probability at least~$.99$ between 
  \begin{equation}
      (1)~\chisquare{\p}{\uniform_\dims} \leq .99\thresh ~~\text{and}~~ \norminf{\p} \leq 100/\dims, 
      \qquad \textrm{and} \qquad (2)~\hellingersq{\p}{\uniform_\dims} > \thresh.    
  \end{equation}
\end{theorem}
We remark that most prior works on uniformity testing~\cite{BHH11,CFMdW10,GL20,LuoWL24} also had some additional such fine-grained aspects, beyond what is stated in \Cref{table:1}.

\paragraph{Additional results.}  Speaking of $\chi^2$-divergence, we mention two additional results we prove at the end of our work. These results additionally inform our \Cref{conj:tv}.

First, as mentioned earlier, in \Cref{sec:small} we give an alternative proof of the $O(\dims^{1/2}/\dst)$ upper bound of~\cite{LuoWL24}, and---like in that work---our result is tolerant with respect to $\chi^2$-divergence.  
That is, we prove the strictly stronger result that for $\thresh \leq 1/\dims$, one can use the code $O(\dims^{1/2}/\thresh^{1/2})$ times to distinguish $\chisquare{\p}{\uniform_\dims} \leq c \thresh$ from $\chisquare{\p}{\uniform_\dims} > \thresh$ (for some constant $c > 0$).

Second, recall that $\chisquare{\p}{\uniform_\dims}$ can be as large as~$\dims$.  
For example, $\chisquare{\uniform_{\cyan{S} }}{\uniform_\dims} = \frac{\dims}{\cyan{r} } - 1$ for any set $\cyan{S} \subseteq [\dims]$ of size $\cyan{r}$. Thus it makes sense to consider the uniformity testing problem even with respect to a $\chi^2$-divergence threshold $\thresh$ that exceeds~$1$.  
In \Cref{sec:giant} we show (albeit only in the quantum string oracle model) that for $\thresh \geq 1$, one can use the code $O(\dims^{1/3}/\thresh^{1/3})$ times to distinguish $\chisquare{\p}{\uniform_\dims} \leq c \thresh$ from $\chisquare{\p}{\uniform_\dims} > \thresh$, and this is optimal.

\section{Technical overview of our proof}\label{sec:ideas}
Our main algorithm is concerned with achieving the best possible $\dst$-dependence for uniformity testing while maintaining a $\dims$-dependence of~$\dims^{1/3}$; in this way, it is best compared with the older works of~\cite{BHH11,CFMdW10}, the latter of which achieves complexity~$O(\dims^{1/3}/\dst^2)$, as well as the classical (no-source-code) algorithm achieving complexity $O(\dims^{1/2}/\dst^2)$.
In fact, all four algorithms here are almost the same (except in terms of the number of samples they use).
Let us describe our viewpoint on this common methodology. 

\noindent We consider the algorithm as being divided into two Phases, and we may as well assume each Phase uses $\ns$ samples.
Phase~1 will have two properties:
\begin{itemize}
 \item It will make $\ns$ \emph{black-box} draws from~$\p$ (\ie the source code is not used in Phase~1).

 \item Using these draws, Phase~1 will end by constructing a certain ``random variable''---in the technical sense of a function $Y : [\dims] \to \R$.

  \item The mean of this random variable~$Y$, vis-a-vis the unknown distribution~$\p$, will ideally be close to $\chisquare{\p}{\uniform} = \dims \cdot \normtwo{\p - \uniform_\dims}^2$. 
  That is, ideally $\mu \coloneqq \E_{\p}[Y] = \sum_{j=1}^\dims \p_j Y(j) \approx \chisquare{\p}{\uniform_\dims}$.
\end{itemize}
\noindent Phase~2 then performs a \emph{mean estimation} algorithm on $Y$ (vis-a-vis~$\p$) to get an estimate of~$\mu$ and therefore of~$\chisquare{\p}{\uniform_\dims}$.  
Ideally, the resulting overall algorithm is not just a uniformity tester, but a $\chi^2$-divergence-from-uniformity \emph{estimator}.  
This could then be weakened to a TV-distance uniformity tester using the inequality $\totalvardist{\p}{\uniform_\dims}^2 \leq \chisquare{\p}{\uniform_\dims}$.

The mean estimation algorithm used in Phase~2 differs depending on whether one has the source code or not.  
In the classical (no source code) model, one simply uses the naive mean estimation algorithm based on $\ns$ more black-box samples; by Chebyshev's inequality, this will (with high probability) give an estimate of~$\mu$ to within $\pm O(\sigma/\ns^{1/2})$, where $\sigma \coloneqq \stddev_\p[Y] = \sqrt{\sum_{j=1}^\dims (Y(j) - \mu)^2}$.
In the case of a quantum tester with the source code access, we can use a \emph{Quantum Mean Estimation} (QME) routine; in particular, the one from~\cite{KO23} will (with high probability) yield an estimate of~$\mu$ to within $\pm O(\sigma/\ns)$.\footnote{This QME routine was not available at the time of~\cite{BHH11,CFMdW10} which had to make do with Quantum Approximate Counting~\cite{BHT98}---essentially, QME for Bernoulli random variables.  But this is not the source of our improvement; one can obtain our main theorem with only a $(\polylog \dims)$-factor loss using just Quantum Approximate Counting.} 

A subtle aspect of this overall plan is that the mean~$\mu$ and standard deviation~$\sigma$ of~$Y$ \emph{are themselves random variables} (in the usual sense), where the randomness comes from Phase~1.  
Thus it is natural to analyze $\E_{\text{Phase~1}}[\mu]$ and $\E_{\text{Phase~1}}[\sigma]$.
Of course, these depend on the definition of~$Y$, which we now reveal:  $Y(j) = \frac{\dims}{\ns}X_j - 1$, where $X_j$ denotes the number of times $j \in [\dims]$ was drawn in Phase~1.
The point of this definition of $Y$ is that a short calculation implies 
\begin{equation}
  \E_{\text{Phase~1}}[\mu] = \chisquare{\p}{\uniform_\dims};
\end{equation}
that is, the random variable $\mu$ is an \emph{unbiased estimator} for our quantity of interest, the $\chi^2$-divergence of~$\p$ from~$\uniform_\dims$.
This is excellent, because although the algorithm does not see~$\mu$ at the end of Phase~1, it will likely get a good estimate of it at the end of Phase~2\dots so long as (the random variable) $\sigma$ is small. 

We therefore finally have two sources of uncertainty about our final error (in estimating $\chisquare{\p}{\uniform_\dims}$):
\begin{enumerate}
  \item Although $\E_{\text{Phase~1}}[\mu] = \chisquare{\p}{\uniform_\dims}$, the random variable $\mu$ may have fluctuated around its expectation at the end of Phase~1.  One way to control this would be to bound $\var_{\text{Phase~1}}[\mu]$ (and then use Chebyshev).
  \item The Phase~2 mean estimation incurs an error proportional to~$\sigma$. One way to control this would be to bound $\E_{\text{Phase~1}}[\sigma^2]$ (and then use Markov to get a high-probability bound on~$\sigma^2$, and hence~$\sigma$).
\end{enumerate}
The quantities controlling the error here, $\var_{\text{Phase~1}}[\mu]$ and $\E_{\text{Phase~1}}[\sigma^2]$, are explicitly calculable symmetric polynomials in $\p_1, \dots, \p_\dims$ of degree at most~$4$, depending on~$\ns$.
In principle, then, one can relate these quantities to $\chisquare{\p}{\uniform_\dims} = \dims \cdot \normtwo{\p-\uniform_\dims}^2$ itself, and derive a bound on how big~$\ns$ must be to (with high probability) get a good estimate of $\chisquare{\p}{\uniform_\dims}$. 

In the classical (no source code) case, this methodology is a way to obtain the $O(\dims^{1/2}/\dst^2)$ sample complexity, adding to the number of existing classical sample-optimal algorithms for the task.
(This method in particular has some potential useful applications; \eg  one could consider decoupling the number of samples used in Phases~1 and~2 to, \eg obtain tradeoffs for memory-limited settings).
On one hand, with this method one can give a very compressed proof of the $O(\dims^{1/2}/\dst^2)$ that, factoring out routine calculations, 
fits in half a page (see, \eg~\cite[Sec.~10]{OW23}). 
On the other hand, one has to execute the calculations and estimations with great care, lest one would obtain a suboptimal result (there is a reason it took $8$~years\footnote{Technically, it took more than $8$ years, as the proof of~\cite{Paninski08} was later shown to have a flaw: so the tight dependence had to wait until~\cite{AcharyaDK15}. See~\cite[Section~2.3]{CanonneTopicsDT2022} for a discussion.} to get the optimal quadratic dependence on~$\dst$~\cite{GRexp:00,Paninski08}). 

In the case when source code is available, so that one can use the QME algorithm, how well does this methodology fare?  
On one hand, QME gives a quadratic improvement over naive classical mean estimation, meaning one can try to use signficantly fewer samples in Phase~2.  
But when one balances out the sample complexity between the two Phases, it implies one is using fewer samples in Phase~1, and hence one gets worse concentration of~$\mu$ around its mean in Phase~1. 
So the calcuations become more delicate.

\subsection{Heuristic calculations}
Instead of diving into complex calculations, let's look at some heuristics.
First, let's consider how the algorithm proceeds in the case when $\p$ really is the uniform distribution $\uniform_\dims$.
In this case, as long as we're in a scenario where $\ns \ll \dims^{1/2}$, we will likely get all distinct elements in Phase~1, meaning that $X_j$ will be~$1$ for exactly $\ns$ values of~$j$ and $X_j$ will be~$0$ otherwise.
Then $Y(j)$ will be $\frac{\dims}{\ns}-1$ for $\ns$ values of~$j$ and will be~$-1$ otherwise.
This indeed means $\mu = \E_{\p}[Y] = \frac{1}{\dims} \sum_{j=1}^\dims Y(j) = 0 = \normtwo{\p - \uniform_\dims}$ with \emph{certainty} in Phase~1.  
This is very good; we get no error out of Phase~1.
However QME in Phase~2 will not perfectly return the value~$\mu =0$; rather, it will return something in the range $\pm O(\sigma/\ns)$, where $\sigma = \sqrt{\frac{1}{\dims}\sum_{j=1}^\dims (Y(j)-0)^2} = \sqrt{\frac{\dims}{\ns} - 1} \sim \dims^{1/2}/\ns^{1/2}$.
Thus the value returned by QME may well be around $\dims^{1/2}/\ns^{3/2}$, which from the algorithm's point of view is consistent with $\chisquare{\p}{\uniform_\dims} \approx \dims^{1/2}/\ns^{3/2}$.
Thus the algorithm will only become confident that $\totalvardist{\p}{\uniform_\dims}^2 \lessapprox \dims^{1/2}/\ns^{3/2}$, and hence it can only confidently accept in the case $\p = \uniform_\dims$ provided $\dims^{1/2}/\ns^{3/2} \lessapprox \dst^2$; \ie $\ns \gtrapprox \dims^{1/3}/\dst^{4/3}$.
We thereby see that with this algorithm, a uniformity testing upper bound of $O(\dims^{1/3}/\dst^{4/3})$ is the \emph{best} we can hope for.  
If one also believes that this algorithm might be optimal (and it \emph{has} been the method of choice for essentially all previously known results), then this could possibly be taken as evidence for our \Cref{conj:tv}. 

At this point, one might try to prove that complexity $O(\dims^{1/3}/\dst^{4/3})$ \emph{is} achievable; so far we have only argued that with this many samples, the algorithm will correctly accept when $\p = \uniform_\dims$ (with high probability).
Again, before jumping into calculations, one might try to guess the ``hardest'' kind of $\dst$-far distributions one might face, and try to work out the calculations for these cases.
The hardest distributions in the classical case (\ie the ones that lead to the matching $\Omega(\dims^{1/2}/\dst^2)$ lower bound) are very natural: they are the $\p$'s in which half of the elements $j \in [\dims]$ have $\p_j = \frac{1+2\dst}{\dims}$ and half have $\p_j = \frac{1-2\dst}{\dims}$.  
Assuming this is the ``worst case'', one can calculate what $\var_{\text{Phase~1}}[\mu]$ and $\E_{\text{Phase~1}}[\sigma^2]$ will be, and the calculations turn out just as desired.  
That is, with $\ns = O(\dims^{1/3}/\dst^{4/3})$, these two error quantities can be shown to be suffciently small so that the overall algorthm will correctly become confident that $\chisquare{\p}{\uniform_\dims} = \dims \cdot \normtwo{\p - \uniform_\dims}^2  \leq 4\dims \cdot \totalvardist{\p}{\uniform_\dims}^2$ significantly exceeds~$\dst^2$, and hence the algorithm can correctly reject. 

Everything therefore looks good, but there is a fly in the ointment.  
Even though this particular~$\p$ with its values of $\frac{1 \pm 2\dst}{\dims}$ seems like the ``hardest'' distribution to face, one still has to reason about all possible~$\p$'s with $\totalvardist{\p}{\uniform_\dims}$.
And when one does the calculations of $\var[\mu]$ and $\E[\sigma^2]$ as prescribed by the standard methodology, the plan ends up \emph{failing}.
Specifically one gets too much error for $\p$'s that have somewhat ``heavy'' elements, meaning $\p_j$'s with $\p_j \gg 1/\dims$.   The prior works~\cite{BHH11,CFMdW10} cope with this failure by taking more samples; \ie setting $\ns = O(\dims^{1/3}/\dst^{c})$ for $c > 4/3$ (specifically, \cite{CFMdW10} achieves $c=2$).  But our goal is to show that this is unnecessary---that the algorithm itself works, even though the standard and natural way of analyzing it fails. 

In short, the reason the standard analysis of the algorithm fails is due to ``rare events'' that are caused by heavy elements in~$\p$. 
These $j$'s with $\p_j \gg 1/\dims$ may well still have $\p_j \ll 1/\ns$ (for our desired $\ns = O(\dims^{1/3}/\dst^{4/3})$), and thus be drawn only rarely in Phase~1.  
The major difficulty is that when they \emph{are} drawn, they generate a \emph{very} large contribution to~$\sigma^2$, causing $\E_{\text{Phase~1}}[\sigma^2]$ to be ``misleadingly large''.  
That is, when there are heavy elements, $\sigma^2$ may have the property of typically being much smaller than its expectation.
Thus controlling the QME error using the \emph{expected} value of~$\sigma^2$ is a bad strategy. 

Perhaps the key insight in our analysis is to show: \emph{In those rare Phase~1 outcomes when $\sigma^2$ is unusually large, $\mu$ is \emph{also} unusually large compared to its expectation.}  The latter event is helpful, because if $\mu$ ends up much bigger than its expectation, we can tolerate a correspondingly worse error-bar from QME.  In short, we show that the rare \emph{bad} outcomes for $\sigma^2$ coincide with the rare \emph{good} outcomes for $\mu$. 

In order to make this idea work out quantitatively, we (seem to) need to weaken our ambitions and get something a bit worse than a $\chi^2$-divergence-from-uniform estimation algorithm, in two ways. (This is fine, as our main goal is just a non-tolerant uniformity tester with respect to TV.)  
First, rather than insisting that we accept with high probability when $\chisquare{\p}{\uniform_\dims} \leq .99 \thresh$ and reject with high probability when $\chisquare{\p}{\uniform_\dims} > \thresh$, we need to only require rejection when $\hellingersq{\p}{\uniform_\dims} > \thresh$.
The reason is that the rare large values of $\sigma^2$ that we face are only comparable with the larger value $\hellingersq{\p}{\uniform_\dims}$, and not with $\chisquare{\p}{\uniform_\dims}$.\footnote{We remark that this $\chi^2$-versus-Hellinger-squared dichotomy is quite reminsicent of the one that occurs in classical works on identity testing, such as~\cite{AcharyaDK15}.} 

As for the second weakening we need to make:  We explicitly add to our tester a check that the value of $\max_j\{X_j\}$ arising after Phase~1 is not too large.  Roughly speaking, this extra test ensures that there are no very heavy elements.  %
(Of course, this is satisfied when $\p = \uniform_\dims$, so we don't mind adding this test.)
The reason we need to add this check is so that we can bound the quadratic expression $\sum_{j = 1}^\dims X_j^2$ (which enters into the value of~$\sigma^2$) by $\max_j\{X_j\} \cdot \sum_{j=1}^\dims X_j$; in turn, once $\max_j\{X_j\}$ is checked to be small, this expression can be bounded by the linear quantity $\sum_{j=1}^\dims X_j$, which can be related to~$\mu$.
It is by relating $\sigma^2$ to $\mu$ in this way that we are able to show the correlation between rare events --- that when $\sigma^2$ is big, $\mu$ is also big. 

To conclude, we apologize to the reader for writing a ``technical overview'' whose length is nearly comparable to that of the actual proof itself. While we tried to make our argument as streamlined and concise as possible, we felt that it was worth conveying the ideas and detours which led us there, and which, while now hidden, motivated the final proof.

\section{Preliminaries}
  \label{sec:prelims}
\subsection{Probability distances}
Throughout, $\log$ and $\ln$ the binary and natural logarithms, respectively. We identify a probability distribution $\p$ over $[\dims]=\{1,2,\dots,\dims\}$ with its probability mass function (pmf), or, equivalently, a vector $\p\in\R^\dims$ such that $\p_i \geq 0$ for all $i$ and $\sum_{i=1}^\dims \p_i=1$. For a subset $S\subseteq [\dims]$, we accordingly let $\p(S) = \sum_{i\in S}\p_i$. The \emph{total variation distance} between two distributions $\p,\q$ over $[\dims]$ is defined as
\begin{equation}
  \totalvardist{\p}{\q} = \sup_{S\subseteq [\dims]} \{\p(S)-\q(S)\} = \frac{1}{2}\normone{\p-\q} \in [0,1],
\end{equation}
where the last equality is from Scheff\'e's lemma. By Cauchy--Schwarz, this gives us the relation 
\begin{equation}
  \frac{1}{2} \normtwo{\p-\q} \leq \totalvardist{\p}{\q} \leq \frac{\sqrt{\dims}}{2} \normtwo{\p-\q}.
\end{equation}
We will in this paper also consider other notions of distance between probability distributions: the \emph{squared Hellinger distance}, defined as
\begin{equation}
  \hellingersq{\p}{\q} = \sum_{i=1}^\dims \Paren{\sqrt{\p_i}-\sqrt{\q_i}}^2 = \normtwo{\sqrt{\p}-\sqrt{\q}}^2 \in [0,2].
\end{equation}
(Some texts normalize this by a factor of $\frac12$; we do not do so, as it makes our statements cleaner.) 
The \emph{chi-squared divergence} is then defined as
\begin{equation}
  \chisquare{\p}{\q} = \sum_{i=1}^\dims \frac{(\p_i-\q_i)^2}{\q_i} =\parens*{\sum_{i=1}^\dims\frac{\p_i^2}{\q_i}}-1\,,
\end{equation}
while the \emph{Kullback--Leibler divergence} (least relevant to us, but quite common in the literature), in nats, is defined as
\begin{equation}
  \kldiv{\p}{\q} = \sum_{i=1}^\dims \q_i \ln \frac{\q_i}{\p_i}\,.
\end{equation}

\noindent As mentioned in \Cref{eqn:ineqs}, we have the following well known~\cite{GS02} chain of inequalities:

\begin{equation}  %
  \totalvardist{\p}{\q}^2 \leq \hellingersq{\p}{\q} \leq \kldiv{\p}{\q} \leq \chisquare{\p}{\q}.
\end{equation}
Moreover, for the special case of the uniform distribution $\uniformOn{\dims}$ over $[\dims]$, we have
\begin{equation}
\chisquare{\p}{\uniformOn{\dims}} = \dims\cdot \normtwo{\p-\uniformOn{\dims}}^2\,.
\end{equation}

\subsection{Distribution access models}

For a probability distribution $\p$ on $[\dims]$, we say a unitary $U_{\p}$ is a \emph{synthesizer} for $\p$ if for some $k$
\begin{equation}
  U_{\p}\ket{0^k} = \sum_{i\in[\dims]} \sqrt{p_i} \ket{i}\ket{\psi_i},
\end{equation}
where the $\ket{\psi_i}$'s are normalized states often called ``garbage states''. Note that any classical randomized circuit using $S$ gates that samples from $\p$ can be converted to a synthesizer $U_{\p}$ in a purely black-box way with gate complexity $O(S)$. (See \cite{KO23} for details and a more thorough discussion of synthesizers.)

In this paper, we say an algorithm makes $t$ uses of ``the code for $\p$'' to mean that we use (as a black box) the unitaries $U_{\p}$, $U_{\p}^\dagger$, and controlled-$U_{\p}$ a total of $t$ times in the algorithm. Each of these unitaries is easy to construct given an explicit gate decomposition of $U_{\p}$ with the same gate complexity up to constant factors.

The \emph{quantum string oracle}, which is used in many prior works, is a specific type of source code for $\p$. Here we have standard quantum oracle access to an $m$-bit string $x \in [\dims]^m$ for some $m$. For any symbol $i \in [\dims]$, the probability $\p_i$ is defined as the frequency with which that symbol appears in $x$, i.e., $\p_i = \frac{1}{m} \left| \{j:x_j=i\} \right|$. Note that calling this oracle on the uniform superposition over $m$ gives us a synthesizer for $\p$. When a randomized sampler for $\p$ is converted to a synthesizer, we get a quantum string oracle, but quantum string oracles are not as general as arbitrary synthesizers. For example, all probabilities described by a quantum string oracle will be integer multiples of $\frac{1}{m}$, whereas an arbitrary synthesizer has no such constraint.

\subsection{Quantum Mean Estimation}
When we use QME, we will have the source code for some distribution $\p$ on $[\dims]$, and we will also have explicitly constructed some (rational-valued) random variable $Y : [\dims] \to \Q$ (say, simply as a table).  
From this, one can easily generate code that outputs a sample from~$Y$ (i.e., outputs $Y(\bj)$ for $\bj \sim [\dims]$), using the code for~$\p$ just one time.
We will then use the following QME result from~\cite{KO23}:
\begin{theorem} \label{thm:QME}
  There is a computationally efficient quantum algorithm with the following guarantee: Given the source code for a random variable~$Y$, as well as parameters $\ns$ and $\errprob$, the algorithm uses the code $O(\ns \log(1/\errprob))$ times and outputs an estimate $\wh{\bmu}$  such that $\Pr[|\wh{\mu} - \mu| > \sigma/\ns] \leq \errprob$, where $\mu = \E[Y]$ and $\sigma = \stddev[Y]$.
\end{theorem}

\section{Algorithm in the Large Distance Regime}
  \label{sec:large}
\newcommand{\cc}{c}
\newcommand{\cq}{C}
In this section, we establish~\cref{thm:main}, our main technical contribution. 
We do this by proving the strictly stronger \Cref{thm:main2}, which we restate more formally:
\begin{theorem}
	\label{theo:medium:case}
For any constant $B > 0$, there exists a computationally efficient quantum algorithm (\cref{algo:medium:case}) with the following guarantees: on input $\frac{1}{\ab} \leq \thresh \leq 1$, it makes 
$O(\ab^{1/3}/\thresh^{2/3})$  uses (where the hidden constant depends on $B$) of ``the code'' for an unknown probability distribution $\p$ over $[\dims]$, and satisfies
\begin{enumerate}
	\item If $\chisquare{\p}{\uniform_{\ab}} \leq .99 \thresh$ and $\norminf{\p} \leq B/\ab$, then the algorithm will \accept with probability at least $.99$.
	\item If $\hellingersq{\p}{\uniform_{\ab}} \geq \thresh$, %
	then  the algorithm will \reject with probability at least $.99$.%
\end{enumerate}
\end{theorem}

\begin{algorithm}[htbp]
	\begin{algorithmic}[1]
	\Require Parameter $\frac{1}{\ab} \leq \thresh \leq 1$, constant $B \geq 1$.
	\State Let $\cc = \cc(B)$ and let $\cq = \cq(\cc)$ be sufficiently large, and let $L$ be defined as
	\[
	L \coloneqq \begin{cases}
		100 & \text{if $\ns \leq \ab^{.99}/B$,} \\
		B\cc \ln \ab & \text{if $\ns > \ab^{.99}/B$.}
	\end{cases}
	\]
    \State Set $\ns \coloneqq \lceil \cc \ab^{1/3}/\thresh^{2/3}\rceil$.
	\State Make $\ns$ draws $\bJ_1, \dots, \bJ_{\ns}$, and let $\bX_j = \sum_{t=1}^{\ns} \indic{J_t=j}$ be the number of times $j\in[\ab]$ is seen.
    \If{$\bX_j \geq L$ for any $j$} 
		\reject  \label{algo:Linf} 
		\EndIf
	\State Do QME with $\cq\ns$ ``samples'' on the random variable $\bY$ defined by $\bY_j = \frac{\ab}{\ns} X_j - 1$, obtaining $\wh{\bmu}$. 
	\If{$\wh{\bmu} \leq .995 \thresh$} 
		\accept 
	\Else\ 
		\reject
	\EndIf
	\end{algorithmic}
		\caption{\label{algo:medium:case}for the large distance regime}
\end{algorithm}
\begin{proof}
	Let us start by recording the following inequalities that we will frequently use: 
	\begin{equation}	\label{ineq:nbd}
		\ns =  \lceil \cc \ab^{1/3}/\thresh^{2/3}\rceil,\ \thresh \geq 1/\ab \quad\implies\quad \cc/\thresh \leq \ns \leq \cc\ab.
	\end{equation}
	We begin with a simple lemma regarding the check on~\cref{algo:Linf}:
	\begin{lemma} \label{prop:line3}
	If $\norminf{\p} \leq B/\ab$, then \Cref{algo:Linf} will \reject with probability at most~$.001$.
	Conversely, if $\norminf{\p} > 2L/\ns$, then \Cref{algo:Linf} will \reject with probability at least~$.999$.
\end{lemma}
\begin{proof}
	Let $\bX_j \sim \text{Bin}(\ns,p_j)$ denote the number of times~$j$ is drawn.
	The second (``conversely'') part of of the proposition follows from a standard Chernoff bound. %
	As for the first part, suppose $\norminf{\p} \leq B/\ab$.
	Now on one hand, if $\ns \leq \ab^{.99}/B$, so that $L = 100$, we have
	\begin{equation}
		\Pr[\text{Bin}(\ns,p_j) \geq 100] \leq \binom{\ns}{100} p_j^{100} \leq ((e\ns/100)p_j)^{100} \leq (e/(100\ab^{.01}))^{100} \leq .001/\ab,
	\end{equation}
	and thus $\bX_j < 100$ for all~$j$ except with probability at most~$.001$, as desired.
	Otherwise, $L = B\cc \ln \ab$, and since $
		\E[\bX_j] \leq B\ns/\ab \leq B\cc,
	$
	the desired result follows from a standard Chernoff and union bound (provided $\cc$ is large enough).
\end{proof}
	\noindent From this, we conclude:
	\begin{itemize}
		\item In Case (1), Line~\ref{algo:Linf} rejects with probability at most $.001$.
		\item In Case (2), we may assume $\norminf{\p}\leq 2L/\ns$ and $\|\bX\|_\infty \leq L$, else Line~\ref{algo:Linf} rejects with probability $\geq .999$. Call this observation $(\diamondsuit)$.
	\end{itemize}

	Now to begin the QME analysis, write $p_j = \frac{1+\eps_j}{\ab}$, where $\eps_j \in [-1, \ab-1]$, 
	and let $\bmu = \sum_{j=1}^{\ab} p_j \bY_j$, the mean of~$\bY$ (from QME's point of view).
	Writing $\eta \coloneqq \hellingersq{p}{\uniform_{\ab}} $, our first goal will be to show:
	\begin{align}
		\text{In Case (1),\qquad} \bmu & \leq .991\thresh && \text{except with probability at most $.001$;}\label{eqn:case1}\\
		\text{In Case (2),\qquad} \bmu & \geq .997\eta && \text{except with probability at most $.002$.} \label{eqn:case2}
	\end{align}
	Starting with \Cref{eqn:case1}, a short calculation (using $\sum_{j=1}^{\ab} \eps_j = 0$) shows
	\begin{equation} \label{eqn:bmu}
		\bmu = \avg_{t=1}^{\ns} \{\eps_{\bJ_t}\}
	\quad\implies\quad
		\E[\bmu] = \frac{1}{\ab}  \sum_{j=1}^{\ab} \eps_j^2 = \chisquare{p}{\uniform_{\ab}}
	\quad\implies\quad \E[\bmu] \leq .99\thresh \text{ in Case~(1).}
	\end{equation}
	Also in Case~(1) we get from \Cref{eqn:bmu} that
	\begin{equation}	\label{ineq:varmu}
		\var[\bmu] = \frac{1}{\ns} \var_{\bj \sim p}[\eps_{\bj}] \leq \frac{1}{\ns} \E_{\bj \sim p}[\eps_{\bj}^2] \leq \frac{B}{\ns \ab}
		\sum_{j=1}^{\ns} \eps_j^2 = \frac{B}{\ns} \chisquare{p}{\uniform_{\ab}} \leq \frac{.99B\thresh}{\ns} \leq \frac{B\thresh^2}{\cc},
	\end{equation}
	the last inequality using \Cref{ineq:nbd}.
	Combining the preceding two inequalities and using Chebyshev, we indeed conclude \Cref{eqn:case1} (provided $\cc = \cc(B)$ is sufficiently large).

	Towards \Cref{eqn:case2},  let  $b \geq 2$ be a certain universal constant to be chosen later, and say that $j \in [\ab]$ is \emph{light} if $p_j \leq b/\ab$ (i.e., $\eps_j \leq b-1$), \emph{heavy} otherwise. 
	We will write 
	\begin{equation}
		\bmu_1 = \avg_{t=1}^{\ns}\{\eps_{\bJ_t} : \bJ_t \text{ heavy}\} \geq 0, \quad	 \bmu_2 = \avg_{t=1}^{\ns}\{\eps_{\bJ_t} : \bJ_t \text{ light}\} \qquad \text{(so $\bmu = \bmu_1 + \bmu_2$),}
	\end{equation}
	and also observe
	\begin{equation}	\label{eqn:helly}
		\eta =\hellingersq{p}{\uniform_{\ab}} = \frac{1}{\ab} \sum_{j=1}^{\ab}(\sqrt{1+\eps_j} - 1)^2 \leq \frac{1}{\ab} \sum_{j=1}^{\ab}\min\{|\eps_j|, \eps_j^2\} \leq \frac{1}{\ab} \sum_{\text{heavy }j} \eps_j + \frac{1}{\ab} \sum_{\text{light }j}\eps_j^2 \eqqcolon \eta_1 + \eta_2.
	\end{equation}
	Let us now make some estimates.
	 First:
	\begin{equation} \label{eqn:pheavy}
		p_{\text{heavy}} \coloneqq \sum_{j \text{ heavy}} p_j = \frac{1}{\ab}\sum_{j \text{ heavy}} (1+\eps_j) \geq \eta_1.%
	\end{equation}
	Also, similar to our Case~(1) estimates we have
	\begin{equation}	\label{eqn:case2a2}
		\E[\bmu_2] = \frac{1}{\ab} \sum_{\text{light } j} (\eps_j^2 + \eps_j) = \eta_2  - \eta_1 \quad \text{(where we used $\sum_{j=1}^{\ab} \eps_j = 0$),}
	\end{equation}
	and
	\begin{equation} \label{eqn:thevar}
		\var[\bmu_2] = \frac{1}{\ns} \var_{\bj \sim p}[\indicSet{\bj \text{ light}} \cdot \eps_{\bj}] 
		\leq \frac{1}{\ns} \E_{\bj \sim p}[\indicSet{\bj \text{ light}} \cdot \eps_{\bj}^2] 
		\leq \frac{b}{\ns\ab} \sum_{j \text{ light}} \eps_j^2 = \frac{b}{\ns} \eta_2 \leq \frac{b}{\cc} \thresh\eta_2 \leq \frac{b}{\cc}\eta_2 \eta \text{ (in Case~(2)).}
	\end{equation}
	We will now establish \Cref{eqn:case2}; in fact, we we even will show the following very slightly stronger fact:
	\begin{align}	
		\text{In Case (2),\qquad} \bmu & \geq .997(\eta_1+\eta_2) \geq .997\eta && \text{except with probability at most $.002$.} \label{eqn:case22}	
	\end{align}
	We divide into two subcases:
	\paragraph{Case (2a):} $\eta_1 \leq .001\eta_2$.  \quad In this case we have $\eta_2 \geq \frac{1}{1.001} (\eta_1+\eta_2)$, and $\E[\bmu_2] \geq .999\eta_2$ from \Cref{eqn:case2a2}.  Since \Cref{eqn:thevar} implies $\var[\bmu_2] \leq 1.001\frac{b}{\cc} \eta_2^2$, Chebyshev's inequality tells us that $\bmu_2 \geq .998\eta_2$ except with probability at most~$.001$ (provided~$\cc$ is large enough).  
	But then $\bmu \geq \bmu_2 \geq \frac{.998}{1.001}(\eta_1+\eta_2)$, confirming \Cref{eqn:case22}.
	\paragraph{Case (2b):} $\eta_1 > .001\eta_2$.  \quad In this case we have $\eta_1 \geq \frac{.001}{1.001} (\eta_1+\eta_2) \geq .0009(\eta_1+\eta_2)$.  
	We now use that heavy~$j$ have $\eps_j \geq b-1$ to observe that
	\begin{equation}
		\bmu_1 = \avg_{t=1}^{\ns}\{\eps_{\bJ_t} : \bJ_t \text{ heavy}\} \geq (b-1) \cdot (\text{fraction of } \bJ_t\text{'s that are heavy}) =  (b-1) \cdot \frac{\text{Bin}(\ns, p_{\text{heavy}})}{\ns}
	\end{equation}
    (in distribution).
	We see that $\E[\bmu_1] \geq (b-1)p_{\text{heavy}}$, and moreover concentration of Binomials and \Cref{eqn:pheavy} imply that 
	\begin{equation} \label{ineq:crunch1}
		\bmu_1\geq \frac12(b-1)p_{\text{heavy}} \geq \frac12(b-1)\eta_1 \text{ except with probability at most $.001$},	
	\end{equation}
	provided that $p_{\text{heavy}} \ns$ is a sufficiently large constant.
	But we can indeed ensure this by taking $\cc$ sufficient large: by \Cref{eqn:pheavy},  being in Case~(2b), and \Cref{ineq:nbd}, it holds that
	\begin{equation}
		p_{\text{heavy}} \ns \geq \eta_1 \ns \geq .0009(\eta_1+\eta_2) \ns \geq .0009 \eta \ns \geq .0009\thresh \ns \geq .0009\cc.
	\end{equation}
	At the same time, \Cref{eqn:case2a2} certainly implies $\E[\bmu_2] \geq -\eta_1$, and \Cref{eqn:thevar} implies $\var[\bmu_2] \leq \frac{b}{\cc} \eta_2 (\eta_1+\eta_2) \leq \frac{1000\cdot 1001 b}{\cc} \eta_1^2$ (using Case~(2b)).
	Thus Chebyshev implies 
	\begin{equation} \label{ineq:crunch2}
		\bmu_2 \geq -1.1\eta_1 \text{ except with probability at most~$.001$},
	\end{equation}
	provided $\cc$ is large enough.	
	Combining \Cref{ineq:crunch1,ineq:crunch2} yields
	\begin{equation}
		\bmu = \bmu_1 + \bmu_2 \geq (\tfrac{b-1}{2}-1.1)\eta_1 \geq .0009(\tfrac{b-1}{2}-1.1)(\eta_1+\eta_2) \text{ except with probability at most~$.002$},
	\end{equation}
	which verifies \Cref{eqn:case22} provided $b$ is a large enough constant.

	\bigskip

	We have now verified the properties of~$\bmu$ claimed in \Cref{eqn:case1,eqn:case22}.  
	Next we analyze the random variable~$\bsigma^2$ that represents the variance of~$\bY$ (from QME's point of view). 
    Our goal will be to show:
    \begin{align}
		\text{In Case (1),\qquad} \bsigma^2/(\cq\ns)^2 & \leq 10^{-6} \cdot \thresh^2 && \text{except with probability at most $.001$,}  \label{eqn:caseee1}\\
		\text{In Case (2),\qquad} \bsigma^2/(\cq\ns)^2 & \leq 10^{-6} \cdot \bmu^2 && \text{except with probability at most $.001$.} \label{eqn:caseee2}
	\end{align}
    Together with \Cref{eqn:case1,eqn:case22}, these facts are sufficient to complete the proof of the theorem, by the QME guarantee of \Cref{thm:QME}.

    We have:
	\begin{equation} \label{ineq:recall}
		\bsigma^2 \coloneqq \sum_{j=1}^{\ab} p_j\bY_j^2 - \bmu^2 = (\ab/\ns)^2\sum_{j=1}^{\ab} p_j\bX_j^2 - (\bmu+1)^2 \leq  (\ab/\ns)^2\sum_{j=1}^{\ab} p_j\bX_j^2 = \bsigma_S^2 + \bsigma_{S^c}^2,
	\end{equation}
	where we've defined $\bsigma_S^2 \coloneqq (\ab/\ns)^2\sum_{j \in S} p_j\bX_j^2$ and $S^c = [\ab] \setminus S$.  
	We will be making two different choices for~$S$ later, but we will always assume
	\begin{equation}	\label{ineq:neg}
		S \supseteq \{j : j \text{ light}\}, \quad \text{which implies }\sum_{j \in S} \eps_j \leq 0
	\end{equation}
	(the implication because $\sum_{j=1}^{\ab} \eps_j = 0$ and $S^c$ contains only $j$'s with $\eps_j \geq b-1 \geq 0$).
	Now since $\E[\bX_j^2] = \ns p_j(1-p_j) + (\ns p_j)^2 \leq \ns p_j + (\ns p_j)^2$, we have
	\begin{align}
		\E[\bsigma_S^2] &\leq (\ab^2/\ns)\sum_{j\in S} p_j^2 + \ab^2 \sum_{j\in S} p_j^3 \\
		&\leq \ab/\ns  +  (2/\ns)\sum_{j\in S}\eps_j + (1/\ns)\sum_{j\in S}\eps_j^2 + 1/\ab + (3/\ab) \sum_{j \in S} \eps_j +  (3/\ab) \sum_{j\in S} \eps_j^2 + (1/\ab) \sum_{j\in S} \eps_j^3 \label{ineq:just1}\\
        &\leq (5\cc\ab/\ns)\parens*{1+ \frac{1}{\ab}\sum_{j\in S}\eps_j + \frac{1}{\ab}\sum_{j\in S}\eps_j^2} + \frac{1}{\ab} \sum_{j\in S} \eps_j^3 
        \label{ineq:toreturn}
	\end{align}
	(where the last inequality used  $1/\ab \leq \cc/\ns \leq (\cc-1)\ab/\ns$ from \Cref{ineq:nbd}).
	Using \Cref{ineq:neg} to drop the term of \Cref{ineq:toreturn} that's linear in the~$\eps_j$'s, we thereby conclude
	\begin{align}
		\E[\bsigma_S^2/(\cq\ns)^2] \leq \E[\bsigma_S^2/\ns^2] 
        &\leq (5\cc\ab/\ns^3)\parens*{1 + \frac{1}{\ab}\sum_{j\in S}\eps_j^2} + (\ab^{1/2}/\ns^2) \parens*{\frac{1}{\ab}\sum_{j\in S}\eps_j^2}^{3/2}  \\
        &\leq (5\thresh^2/\cc^2)(1 + \eta_S) +  \frac{\thresh^{4/3}}{\cc^2 \ab^{1/6}} \eta_S^{3/2},
	\end{align}
	where $\eta_S \coloneqq \frac{1}{\ab}\sum_{j\in S}\eps_j^2$.
	In Case~(1) we select $S = [\ab]$, so $\eta_S = \chisquare{p}{\uniform_{\ab}} \leq .99\thresh \leq \thresh \leq 1$, and the above bound gives
	\begin{equation}
		\text{Case (1)} \implies \E[\bsigma^2/(\cq\ns)^2] \leq 10\thresh^2/\cc^2 +\frac{\thresh^{17/6}}{\cc^2 \ab^{1/6}}  \leq \cdot 10^{-9} \cdot \thresh^2
	\end{equation}
	(provided $\cc$ is large enough).  Now \Cref{eqn:caseee1} follows by Markov's inequality.

	In Case~(2) we select $S  = \{j : j \text{ light}\}$, so $\eta_S = \eta_2$ and we conclude (using obvious notation)
	\begin{equation} \label{ineq:put}
		\text{Case (2)} \implies \E[\bsigma_{\text{light}}^2/(\cq\ns)^2] \leq (5\thresh^2/\cc^2)(1+\eta_2)+\frac{\thresh^{4/3}}{\cc^2 \ab^{1/6}}  \eta_2^{3/2}
		\leq .4 \cdot 10^{-9} \cdot (\eta_1+\eta_2)^2,
	\end{equation}
	(provided $\cc$ large enough), 
	where we used $\thresh \leq \eta \leq \eta_1 + \eta_2$ and also $\thresh \leq 1$.
	We now complete the bounding of~$\bsigma^2$ in Case~(2) by two different strategies:

	\paragraph{Case (2.i):} $\ns > \ab^{.99}/B$. \quad In this case, $L = B\cc \ln \ab$, and $(\diamondsuit)$ tells us $\norminf{\p} \leq 2L/\ns$, so we have
    \begin{equation}    \label{ineq:pinf}
        \norminf{\p} \leq \frac{2B\cc \ln \ab}{\ns} \leq \frac{2B^2 \cc \ln \ab}{\ab^{.99}}.
    \end{equation}
    Now returning to \Cref{ineq:toreturn}, we get
	\begin{align}
		\E[\bsigma_{\text{heavy}}^2/(\cq\ns)^2] &\leq \frac{5\cc\ab}{\cq^2\ns^3} + \frac{5\cc\ab}{\cq^2\ns^3}\parens*{1 +  \eps_{\text{max}} + \frac{\ns}{5\cc\ab}\eps_{\text{max}}^2} \cdot \frac{1}{\ab}\sum_{j \text{ heavy}}\eps_j\\
		&\leq \frac{5\thresh^2}{(\cq\cc)^2}  + \frac{5\cc\ab^2}{\cq^2\ns^3}\parens*{\norminf{\p} + \frac{\ns}{5 \cc} \cdot \norminf{\p}^2}\eta_1
        \leq \frac{5\thresh^2}{(\cq\cc)^2}  + \frac{14B^6\cc^2 \ln^2 \ab}{\cq^2 \ab^{1.96}}\eta_1, \label{ineq:12}
	\end{align}
    where we used \Cref{ineq:pinf} and $\ns > \ab^{.99}/B$.
	We can again bound the first expression in \Cref{ineq:12} as $\frac{5\thresh^2}{(\cq \cc)^2} \leq 10^{-6} \cdot (\eta_1 + \eta_2)^2$.
	As for the second expression, either $\eta_1 =0$ (there are no heavy~$j$'s) or else $\eta_1 \geq \frac{b-1}{\ab}$ (there is at least one heavy~$j$). 
	In either case, we have $\eta_1 \leq  \frac{\ab}{b-1}\eta_1^2 \leq \ab \eta_1^2$, so we can bound this second expression  by
	\begin{equation}
		\frac{14B^6\cc^2 \ln^2 \ab}{\cq^2 \ab^{.96}}\eta_1^2
        \leq .4 \cdot 10^{-9} \cdot (\eta_1 + \eta_2)^2
	\end{equation}
    where we used $\cq = \cq(\cc)$ sufficiently large (and we could have taken $\cq = 1$ were willing to assume~$\ab$ sufficiently large).
	Putting this bound together with \Cref{ineq:put} we obtain:
	\begin{equation}
		\text{Case~(2.i)} \implies \E[\bsigma/(\cq\ns)^2] \leq .8\cdot 10^{-9}\cdot (\eta_1+\eta_2)^2 \leq \tfrac{.8}{.997} \cdot 10^{-9} \cdot \bmu^2 \leq \cdot 10^{-9} \cdot \bmu^2,
	\end{equation}
    using \Cref{eqn:case22}.  \Cref{eqn:caseee2} now follows (in this Case~(2.i)) by Markov's inequality.

	\paragraph{Case (2.ii):} $\ns \leq \ab^{.99}/B$. \quad In this case we use a different strategy.  
	Recall from \Cref{ineq:recall} that
	\begin{equation}
		\bsigma^2 \leq (\ab/\ns)^2 \sum_{j=1}^{\ab} p_j \bX_j^2 \leq (\ab/\ns)^2 \|\bX\|_\infty \sum_{j=1}^{\ab} p_j \bX_j = (\ab/\ns) \|\bX\|_\infty (1+\bmu).
	\end{equation}
    By $(\diamondsuit)$ we may assume $\|\bX\|_\infty \leq L = 100$, the equality because we are in Case~(2.ii).
	Thus
	\begin{equation}
		\bsigma^2/(\cq\ns)^2 \leq \bsigma^2/\ns^2 \leq 100 (\ab/\ns^3) (1+\bmu) \leq \frac{100 \thresh^2}{\cc^3}+ \frac{100 \thresh^2}{\cc^3}\bmu \leq 10^{-6} \cdot \bmu^2.
	\end{equation}
	(provided $\cc$ large enough), where we used $\thresh \leq \eta \leq \frac{1}{.997}\bmu$ (from \Cref{eqn:case22}) and also $\thresh \leq 1$.
    This verifies \cref{eqn:caseee2} in Case~(2.ii), completing the proof.
\end{proof}

\section{Algorithm in the Small Distance Regime}
  \label{sec:small}
\newcommand{\dstltwo}{\textcolor{ForestGreen}{\tau}}
In this section, we provide an alternative (and arguably simpler) proof of the main result of~\cite{LuoWL24}:
\begin{theorem}
  \label{thm:smalldistance:restated}
There is a computationally efficient quantum algorithm (\cref{algo:qme:hashing}) for uniformity testing with the following guarantees: it takes $O(\dims^{1/2}/\dst)$ ``samples'' from an unknown probability distribution $\p$ over $[\dims]$, and distinguishes with probability at least $2/3$ between (1)~$\chisquare{\p}{\uniform_\dims} \leq \frac{\dst^2}{144}$, and (2)~$\chisquare{\p}{\uniform_\dims} > \dst^2$.
\end{theorem}
This in turn will follow from the more general result on tolerant $\lp[2]$ closeness testing, where one is given access to the source code for \emph{two} unknown probability distributions $\p,\q$ over $[\dims]$, and one seeks to distinguish $\normtwo{\p-\q} \leq c\cdot \dstltwo$ from $\normtwo{\p-\q} \geq \dstltwo$.
\begin{theorem}
  \label{thm:smalldistance:l2}
There is a computationally efficient quantum algorithm (\cref{algo:qme:hashing}) for closeness testing with the following guarantees: it takes $O(1/\dstltwo)$ ``samples'' from two unknown probability distributions $\p,\q$ over $[\dims]$, and distinguishes with probability at least $2/3$ between (1)~$\normtwo{\p-\q} \leq \frac{\dstltwo}{12}$, and (2)~$\normtwo{\p-\q} > \dstltwo$.
\end{theorem}
\cref{thm:smalldistance:restated} can then be obtained as a direct corollary by setting $\dstltwo = \dst/\sqrt{\dims}$, recalling that when $\q$ is the uniform distribution $\uniformOn{\ab}$, $\lp[2]$ distance and $\chi^2$ divergence are equivalent:
\[
 \normtwo{\p-\uniformOn{\ab}}^2 = \sum_{i=1}^{\dims} (\p_i-1/\dims)^2 = \frac{1}{\dims} \sum_{i=1}^{\dims} \frac{(\p_i-1/\dims)^2}{1/\dims}
 = \frac{1}{\dims} \chisquare{\p}{\uniformOn{\ab}}
\]
We emphasize that the result of~\cref{thm:smalldistance:l2} itself is not new, as a quantum algorithm achieving the same sample complexity (in the same access model) was recently obtained by~\cite{LuoWL24}.\footnote{Technically,~\cite{LuoWL24}'s result can be seen as slightly stronger, in that it allows to test $\normtwo{\p-\q} \leq (1-\gamma)\dstltwo$ vs. $\normtwo{\p-\q}{\uniform_\dims} > \dstltwo$, for arbitrarily small constant $\gamma>0$.} However, our algorithm differs significantly from the one in~\cite{LuoWL24}, and we believe it to be of independent interest for several reasons:
\begin{itemize}
  \item it is \emph{conceptually very simple}: (classically) hash the domain down to \emph{two} elements, and use QME to estimate the bias of the resulting Bernoulli;
  \item it neatly \emph{separates the quantum and classical aspects} of the task, only using QME (as a blackbox) in a single step of the algorithm;
  \item in contrast to the algorithm of~\cite{LuoWL24}, it \emph{decouples the use of the source code from $\p$ and $\q$}, allowing one to run our algorithm when the accesses to the two distributions are on different machines, locations, or even will be granted at different points in time (\ie one can run part of the algorithm using the source code for $\p$, and, one continent and a year apart, run the remaining part on the now-available source code for $\q$ without needing $\p$ anymore).
\end{itemize}
The idea behind~\cref{thm:smalldistance:l2} is relatively simple: previous work (in the classical setting) showed that hashing the domain from $\dims$ to a much smaller $\ab' \ll \dims$ could yield sample-optimal testing algorithms in some settings, \eg when testing under privacy bandwidth, or memory constraints. Indeed, while this ``domain compression'' reduces the total variation distance by a factor $\Theta(\sqrt{\ab'/\dims})$, this shrinkage is, in these settings, balanced by the reduction in domain size. The key insight in our algorithm is then to (1)~use this hashing with respect to $\lp[2]$ distance, not total variation distance, and show that one can in this case get a two-sided guarantee in the distance (low-distortion embedding) instead of a one-sided one; and (2)~compress the domain all the way to $\dims'=2$, so that one can then invoke the QME algorithm to simply estimate the bias of a coin to an additive $\pm \dstltwo$, a task for which a quantum quadratic speedup is well known.
\begin{proof}[Proof of~\cref{thm:smalldistance:l2}]
As mentioned above, a key building block of our algorithm is the following ``binary hashing lemma,'' a simple case of the domain compression primitive of~\cite{ACHST:20}:
\begin{lemma}[Random Binary Hashing (Lemma~2.9 and Remark~2.4 of~\cite{CanonneTopicsDT2022}]
  \label{lemma:random:binary:hashing}
  Let $\p,\q\in\distribs{\ab}$. Then, for every $\alpha\in[0,1/2]$,
  \[
        \probaDistrOf{S}{ |\p(S)-\q(S)| \geq \alpha\normtwo{\p-\q} } \geq \frac{1}{12}(1-4\alpha^2)^2\,,
  \] 
  where $S\subseteq[\ab]$ is a uniformly random subset of $[\ab]$.
\end{lemma}
Given our goal of tolerant testing, we also require a converse to~\cref{lemma:random:binary:hashing}, stated and proven below:
\begin{lemma}
  \label{lemma:random:binary:hashing:2}
  Let $\p,\q\in\distribs{\ab}$. Then, for every $\beta\in[1/2,\infty)$,
  \[
        \probaDistrOf{S}{ |\p(S)-\q(S)| \geq \beta\normtwo{\p-\q} } \leq \frac{1}{4\beta^2}\,,
  \] 
  where $S\subseteq[\ab]$ is a uniformly random subset of $[\ab]$.
\end{lemma}
\begin{proof}
As in the proof of~\cref{lemma:random:binary:hashing}, we write $\delta \eqdef \p-\q\in\R^{\ab}$ and $\p(S)-\q(S) = \frac{1}{2}Z$, where $Z \eqdef \sum_{i=1}^{\ab} \delta_i \xi_i$ for $\xi_1,\dots,\xi_{\ab}$ \iid Rademacher. We will use the following fact established in the proof of this lemma, which we reproduce for completeness: 
\begin{equation}
	\label{eq:random:binary:hashing:expectation}
  \bEE{Z^2} = \sum_{1\leq i,j\leq \ab} \delta_i\delta_j \bEE{\xi_i\xi_j} = \sum_{i=1}^{\ab} \delta_i^2 = \normtwo{\delta}^2\,.
\end{equation}
By Markov's inequality, we then have
\[
        \probaDistrOf{S}{ |\p(S)-\q(S)| > \beta\normtwo{\p-\q} } 
        = \probaDistrOf{S}{ Z^2 > 4\beta^2\bEE{Z^2} } 
        \leq \frac{1}{4\beta^2}
 \]
 concluding the proof.
\end{proof}
While the above two lemmas allow us to obtain a slightly more general result than in the theorem statement by keeping $\alpha,\beta$ as free parameters, for concreteness, set $\alpha \coloneqq 1/(2\sqrt{2})$ and $\beta=4$. This implies the following:
\begin{itemize}
  \item If $\normtwo{\p-\q} \geq \dstltwo$, then
  \[
      \probaDistrOf{S}{ \abs{\p(S)-\frac{|S|}{\dims}} \geq \frac{\dstltwo}{\sqrt{8}} } \geq \frac{1}{48}
  \]
  \item If $\normtwo{\p-\q} \leq \frac{\dstltwo}{12}$, then
  \[
    \probaDistrOf{S}{ \abs{\p(S)-\frac{|S|}{\dims}} \geq \frac{\dstltwo}{\sqrt{9}} } \leq \frac{1}{64}\,.
  \]
\end{itemize}
where $S\subseteq[\ab]$ is a uniformly random subset of $[\ab]$. This allows us to distinguish between the two cases with only $\bigO{1}$ repetitions:
\begin{algorithm}[H]
\begin{algorithmic}[1]
	\State Set $T = O(1)$, $\errprob \eqdef \frac{1}{600}$, $\tau \eqdef \frac{1/48+1/64}{2}$.	\Comment{$\errprob \leq \frac{1}{3}\Paren{\frac{1}{48}-\frac{1}{64}}$.}
	\For{$t = 1$ \textbf{to} $T$}
		\State Pick a u.a.r. subset $S_t\subseteq[\ab]$ (independently of previous iterations)
		\State Estimate $\p(S_t), \q(S_t)$ by $\hat{p}_t,\hat{q}_t$ to within $\pm \frac{\tau}{100}$ with error probability $\errprob$. \Comment{QME}\label{step:qme}
		\If{$|\hat{p}_t-\hat{q}_t| \leq \frac{\dst}{\sqrt{8\dims}}$}
		    \ $b_t \gets 0$
		\Else
		    \ $b_t \gets 1$
		\EndIf
	\EndFor
	\Return \accept if $\frac{1}{T}\sum_{t=1}^T b_t \leq \tau$ \Comment{Estimate of the probability \accept}
\end{algorithmic}
	\caption{QME+Binary Hashing Tester}\label{algo:qme:hashing}
\end{algorithm}
A standard analysis shows that, for $T$ a sufficiently large constant, with probability at least $2/3$ the estimate $\frac{1}{T}\sum_{t=1}^T b_t$ will be within an additive $\errprob+\frac{1}{1000}$ of the corresponding value (either $1/48$ or $1/64$), in which case the output is correct. The total number of samples required is $T$ times the sample of the Quantum Mean Estimation call on Line~\ref{step:qme}, which is $O(1/\dstltwo)$: the complexity of  getting a $O(\dstltwo)$-additive estimate of the mean of a Bernoulli random variable with high (constant) probability. 
This concludes the proof.
\end{proof}

\section{Algorithm in the Giant Distance Regime}
  \label{sec:giant}
\newcommand{\pow}{\mathrm{pow}}
\newcommand{\del}{\delta}
\newcommand{\bigns}{{\textcolor{red}{N}}}

In this final section, we show that, in the (stronger) quantum string oracle model, one can perform tolerant uniformity testing with respect to $\chi^2$ divergence in the ``very large parameter regime,'' that is, to distinguish $\chisquare{\p}{\uniform_\dims} \leq c \thresh$ from $\chisquare{\p}{\uniform_\dims} > \thresh$ for $\thresh \geq 1$: 
\begin{theorem}\label{thm:giant}
  There is a computationally efficient quantum algorithm for uniformity testing with the following guarantees: For $\thresh\geq 1$, the algorithm makes $O(\dims^{1/3}/\thresh^{1/3})$ calls to the quantum string oracle for an unknown distribution $\p$ over $[\dims]$, and distinguishes with probability at least~$.99$ between 
  \begin{equation}
      (1)~\chisquare{\p}{\uniform_\dims} \leq c\cdot \thresh, 
      \qquad \textrm{and} \qquad (2)~\chisquare{\p}{\uniform_\dims} > \thresh\,,    
  \end{equation}
  where $c>0$ is an absolute constant. Moreover, this query complexity is optimal.
\end{theorem}
Note that, as discussed in the introduction, this result does not imply anything in terms of total variation distance, as the latter is always at most $1$; however, we believe this result to be of interest for at least two reasons: (1)~it is in itself a reasonable (and often useful) testing question, when total variation distance is not the most relevant distance measure, and implies, for instance, testing $\chisquare{\p}{\uniform_\dims} \leq c\cdot \thresh$ from $\kldiv{\p}{\uniform_\dims} > \thresh$; and (2)~one can show that this complexity is tight, by a reduction to the $\theta$-to-1 collision problem, which provides additional evidence for~\Cref{conj:tv}. 
\begin{proof}
The main ingredient of the proof is the following lemma, which guarantees that taking $\bigns = \Theta({\dims}/{\thresh})$ from the unknown distribution $\p$ is enough to obtain (with high constant probability) a multiset of elements with, in one case, no collisions, and in the other at least one collision:
\begin{lemma}\label{lem:distinctelements}
For $\thresh \geq 1$, there exists a constant $c\in(0,1)$ such that taking $\bigns$ \iid samples from an unknown $\p$ over $[\dims]$ results in a multiset $S$ satisfying the following with probability at least $.99$:
\begin{itemize}
	\item If $\chisquare{\p}{\uniformOn{\dims}} \leq c\cdot \thresh$, then all elements in $S$ are distinct;
	\item If $\chisquare{\p}{\uniformOn{\dims}} \geq \thresh$, then at least two elements in $S$ are identical;
\end{itemize}
as long as $1601\cdot \frac{\dims}{\thresh} \leq \bigns \leq \frac{1}{10c}\cdot \frac{\dims}{\thresh}$. (In particular, taking $c \eqdef \frac{1}{16010}$ suffices to ensure such a choice of $\bigns$ is possible.)
\end{lemma}
Before proving this lemma, we describe how it implies our stated complexity upper bound. \Cref{lem:distinctelements} guarantees that we can reduce our testing problem to that of deciding if, given oracle access to a string of size $\bigns = \Theta(\sqrt{\dims/\thresh})$, whether all the elements in it are distinct. This problem is solved by Ambainis' element distinctness quantum-walk algorithm~\cite{Amb07} using $O(\bigns^{2/3}) = O({\dims^{1/3}/\thresh^{1/3}})$ quantum queries.

\begin{proof}[Proof of~\Cref{lem:distinctelements}]
Suppose we take $\bigns$ \iid samples $X_1,\dots, X_\bigns$ from $\p$, and count the number $Z$ of collisions among them:
\[
		Z \eqdef \sum_{1\leq i < j \leq \bigns} \indic{X_i = X_j}
\]
Letting $\del \eqdef \p - \uniformOn{\dims}$ and $\pow_t(x) \eqdef \sum_{i=1}^{\dims} x_i^t$ for all integer $t\geq 0$ and vector $x\in\R^{\dims}$ (so that $\del_i = \p_i - 1/\dims$ for all $i$), we have, $\pow_1(\del) = 0$, and 
\[
\pow_2(\del) = \normtwo{\p-\uniformOn{\dims}}^2 = \frac{1}{\dims}\chisquare{\p}{\uniformOn{\dims}}
\]
Now, it is not hard to verify that $\bEE{Z} = \binom{\bigns}{2} \normtwo{\p}^2 = \binom{\bigns}{2} (\pow_2(\del)+1/\ab)$, and 
\begin{align} %
\var[Z] 
&= \binom{\bigns}{2} \normtwo{\p}^2 \Paren{1-\normtwo{\p}^2} + 6\binom{\bigns}{3}\Paren{\norm{\p}_3^3 - \norm{\p}_2^4} \notag\\
&\leq \bEE{Z} + 6\binom{\bigns}{3}\Paren{\pow_3(\del) + \frac{3}{\dims}\pow_2(\del)} \label{eq:variance:collisions}  %
\end{align}
From this, we get, setting $\tau \eqdef \sqrt{{\thresh}/{\dims}} \geq 1/\sqrt{\dims}$:
\begin{itemize}
\item If $\chisquare{\p}{\uniformOn{\dims}}\leq c\cdot \thresh$, then $\pow_2(\del) \leq c^2\cdot \tau^2$, and as long as $\bigns \leq \frac{1}{10c\tau}$ we have $\binom{\bigns}{2} (c^2\cdot \tau^2+1/\ab) \leq 1/100$, so that by Markov's inequality
\[
	\probaOf{Z \geq 1 } \leq \probaOf{Z \geq 100\bEE{Z} } \leq \frac{1}{100}
\]
\item If $\chisquare{\p}{\uniformOn{\dims}}\geq \thresh$, then $\pow_2(\del) \geq \tau^2$,  and by Chebyshev's inequality and~\cref{eq:variance:collisions}
\begin{align*}
	\probaOf{Z = 0 } 
	&\leq \probaOf{\abs{Z - \bEE{Z}} \geq \bEE{Z} } 
	\leq \frac{1}{\bEE{Z}} + \frac{4}{\bigns}\cdot \frac{\pow_3(\del) + \frac{3}{\dims}\pow_2(\del)}{(\pow_2(\del)+1/\ab)^2} \\
	&\leq \frac{2}{\bigns(\bigns-1)\tau^2} + \frac{4}{\bigns}\cdot \frac{\pow_2(\del)^{3/2} + \frac{3}{\dims}\pow_2(\del)}{\pow_2(\del)^2} \\
	&\leq \frac{3}{\bigns^2\tau^2} + \frac{4}{\bigns\tau}+ \frac{12}{\bigns\ab\tau^2}\\
	&\leq \frac{3}{\bigns^2\tau^2} + \frac{4}{\bigns\tau}+ \frac{12}{\bigns} \tag{$\tau\geq 1/\sqrt{\dims}$}\\
	&\leq \frac{3}{\bigns^2\tau^2} + \frac{16}{\bigns\tau}
\end{align*}
which is at most $\frac{1}{100}$ for $\bigns \geq \frac{1601}{\tau}$.
\end{itemize}
This proves the lemma.
\end{proof}
This concludes the proof of the upper bound part of~\Cref{thm:giant}. To conclude, it only remains to show that this is, indeed, optimal. 
For this, we need a lower bound of \cite{Kut05}, which generalized a lower bound of Aaronson and Shi~\cite{AS04}:
\begin{theorem}[\cite{Kut05}]
	Let $\dims>0$ and $\cyan{r}\geq 2$ be integers such that $\cyan{r}|\dims$, and let $f:[\dims]\to[\dims]$ be a function to which we have quantum oracle access. Then deciding if $f$ is 1-to-1 or $\cyan{r}$-to-1, promised that one of these holds, requires $\Omega((\dims/\cyan{r})^{1/3})$ quantum queries.
\end{theorem}
When we view this function as a quantum string oracle for a probability distribution, the function being 1-to-1 corresponds to the uniform distribution on $[\dims]$. In the other case, the distribution is uniform on a subset of size $[\dims/\cyan{r}]$, for any $\cyan{r}\geq \thresh+1$ dividing $\dims$. An easy calculation shows that the second distribution is at $\chi^2$ divergence
\begin{equation}
	\chisquare{\p}{\uniformOn{\dims}} 
	= \sum_{i \in [\dims]} \left(\frac{\p_i^2}{1/\dims}\right) - 1
	= \dims\cdot  \frac{\cyan{r}^2}{\dims^2} \cdot \frac{\dims}{\cyan{r}}	 - 1
	= \cyan{r} -1
	\geq \thresh,
\end{equation}
from uniform, which completes the proof.
\end{proof}
 
\printbibliography

\appendix
\section{Reduction from Identity to Uniformity Testing}
  \label{app:identity:reduction}
As mentioned in the introduction, there is a known reduction from identity to uniformity testing, due to Goldreich~\cite{Goldreich:16} and inspired by~\cite{DK:16}: which, in a blackbox way, converts an instance of uniformity testing (in total variation distance) with reference distribution $\q$ over $[\dims]$ and distance parameter $\dst$ to an instance of uniformity testing over $[4\dims]$ and distance parameter $\dst/4$. (Here, we follow the exposition and parameter setting of~\cite[Section~2.2.3]{CanonneTopicsDT2022}.)

To be able to use it in our setting, all we need to check is that this blackbox reduction $\Phi_\q$ preserves access to ``the code'': that is, given the code $C_\p$ for a probability distribution $\p$ over $[\ab]$, that we can efficiently have access to the code $C_{\p'}$ for the resulting distribution $\p'=\Phi_\q(\p)$ over $[4\dims]$. 
To do so, note that $\Phi_\q$ is the composition of 3 successive mappings,
\[
  \Phi_\q = \Phi_\q^{(1)}\circ\Phi_\q^{(2)}\circ\Phi_\q^{(3)}
\]
where $\Phi_\q^{(3)}\colon[\ab]\to[\ab]$, $\Phi_\q^{(2)}\colon[\ab]\to[\ab+1]$, and , $\Phi_\q^{(2)}\colon[\ab+1]\to[4\ab]$. So it suffices to show that each of these 3 mappings does preserve access to the code generating a sample from the resulting distribution.
\begin{itemize}
  \item The first, $\Phi_\q^{(3)}$, is the easier, as it consists only in mixing its input with the uniform distribution:
  \[
    \Phi_\q^{(3)}(\p) = \frac{1}{2}\p + \frac{1}{2}\uniformOn{\dims}
  \]
  for which a circuit can be easily obtained, given a circuit for $\p$.\vspace{-0.5em}
  \item The second, $\Phi_\q^{(2)}$, ``rounds down'' the probability of each of the $\dims$ elements of the domain, and sends the remaining probability mass to a $(\dims+1)$-th new element:
  \[
      \Phi_\q^{(2)}(\p)_i = 
      \begin{cases}
        \frac{\flr{4\dims\q_i}}{4\dims\q_i}\cdot \p_i, & i\in[\dims]\\
        1- \sum_{i=1}^\dims \frac{\flr{4\dims\q_i}}{4\dims\q_i}\cdot \p_i, & i=\dims+1
      \end{cases}
  \]
  This corresponds to adding to the circuit $C_\p$ for $\p$ a ``postprocessing circuit'' which, if the output of $C_\p$ is $i$, outputs $i$ with probability $\frac{\flr{4\dims\q_i}}{4\dims\q_i}$ (and $\dims+1$ otherwise).\vspace{-0.5em}
  \item The third, $\Phi_\q^{(1)}$, assumes that the reference distribution $\q$ is ``grained'' (namely, all its probabilities are positive multiples of $1/(4\dims)$), which will be the case after the first two mappings\footnote{Specifically, when chaining the three mappings, the reference distribution called $\q$ here is actually $\Phi_\q^{(2)}\circ \Phi_\q^{(3)}(\q)$.} fully known). Having partitioned $[4\dims]$ in sets $S_1,\dots, S_\dims$ where 
  \[
    |S_i| = 4\dims\cdot \q_i \geq 1
  \]
  and $\Phi_\q^{(1)}$ is given by
  \[
      \Phi_\q^{(3)}(\p)_i = \sum_{j=1}^\dims \frac{\p_i}{|S_i|} \indic{j\in S_i} ,\qquad  i\in[4\dims]\,.
  \]
  This corresponds to adding to the circuit $C_\p$ for $\p$ a ``postprocessing circuit'' which, if the output of $C_\p$ is $i$, outputs an element of $S_i$ uniformly at random. (Importantly, $S_1,\dots, S_\dims$ are uniquely determined by $\q$, and do not depend on $\p$ or $C_\p$ at all.)
\end{itemize}
To summarize, each of these three mappings can be implemented to provide, given a circuit $C_\p$ for $\p$, a circuit $C_{\p'}$ for the output $\p'$, so that altogether the reduction can be implemented in a way which preserves access to ``the code.''
 \end{document}